\documentclass[11pt]{article}
\usepackage{amsmath,amssymb,amsthm}
\newtheorem{prop}{Proposition}
\newtheorem{thm}[prop]{Theorem}

\theoremstyle{definition}

\def\theequation{\thesection.\arabic{equation}}
\makeatletter
\@addtoreset{equation}{section}
\makeatother

\title{Hirota bilinear approach to GUE, NLS, and Painlev\'e IV}
\author{Saburo Kakei\\[2mm]
{\small Department of Mathematics, Rikkyo University,}\\
{\small 3-34-1 Nishi-ikebukuro, Toshima-ku, Tokyo 171-8501, Japan}}
\date{}
\begin{document}
\maketitle
\begin{abstract}
Tracy and Widom showed that the level spacing function of the 
Gaussian unitary ensemble is related to 
a particular solution of the fourth Painlev\'e equation. 
We reconsider this problem from the viewpoint of Hirota's bilinear method 
in soliton theory and present another proof.
We also consider the asymptotic behavior of the level spacing 
function as $s\to\infty$, and its relation to the 
``Clarkson-McLeod solution'' to the Painlev\'e IV equation.
\end{abstract}
\hfill
\textit{To the memory of Professor Ryogo Hirota}

\section{Introduction}
Random matrix theory has its origins in mathematical statistics and 
nuclear physics \cite{MehtaBook}, and has been applied to many fields 
such as wireless communications \cite{TulinoVerdu}. 
In the paper \cite{TracyWidom1}, Tracy and Widom introduced a 
limiting distribution function associated with the largest eigenvalue
of the Gaussian unitary ensemble (GUE). The Tracy-Widom distribution 
has been applied to many branch of mathematics such as combinatorics
and statistics. (For review, see \cite{TracyWidomICM}.) 
Furthermore, recent experimental studies reveal that the Tracy-Widom distribution 
appears in various physical systems 
such as slow combustion of paper \cite{MMMT}, 
turbulent liquid crystals \cite{TakeuchiSano}, 
fiber laser systems \cite{FPNFD}, 
evaporating drops of colloidal suspensions \cite{YLSBDY}.
(For a survey of recent experimental studies, see \cite{TakeuchiReview}.) 

The GUE-Tracy-Widom distribution function $F_2(s)$ is given by
\begin{equation}
F_2(s) = \exp\left[-\int_s^{\infty}(x-s)q^2(x)dx\right],
\end{equation}
where $q(x)$ is the unique solution to a special case of the
Painlev\'e II equation (``Hastings-McLeod solution'' \cite{HM}) 
\begin{equation}
\frac{d^2 q}{ds^2}= s\frac{dq}{ds}+2q^3
\end{equation}
satisfying the boundary condition
$q(s) \sim \mathrm{Ai}(s)$ as $s\to\infty$. 
The function $F_2(s)$ is a scaling limit of 
the distribution function of the largest eigenvalue 
of the Gaussian unitary ensemble (GUE). 
For the ensemble of $N\times N$ random Hermitian matrices, 
joint distribution of eigenvalues is given by \cite{MehtaBook}
\begin{equation}
\mathbb{P}_{2,N}(x_1,x_2,\ldots,x_N)=
\frac{1}{Z_N}
\prod_{1\leq i<j\leq N}\left|x_i-x_j\right|^2\prod_{i=1}^N
e^{-x_i^2},
\end{equation}
where the normalization constant $1/Z_N$ is given by
\begin{align}
Z_N &=
\int_{-\infty}^{\infty}\cdots\int_{-\infty}^{\infty}
\prod_{1\leq i<j\leq N}\left|x_i-x_j\right|^2\prod_{i=1}^N
e^{-x_i^2}dx_1\cdots dx_N
\nonumber\\
&= 2^{-N(N-1)/2}\pi ^{N/2}\prod_{j=1}^N j!.
\label{normalization}
\end{align}
We will give a bilinear-theoretical proof of the identity
\eqref{normalization} in Appendix A.

The distribution function of the largest eigenvalue 
$\lambda_{\mathrm{max}}$ is then expressed as
\begin{align}
F_{2,N}(s)  &:=\mathbb{P}_{2,N}(\lambda_{\mathrm{max}}\leq s)
=\mathop{\idotsint}_{x_1\leq \cdots\leq x_N\leq s}\mathbb{P}_{2,N}(x_1,x_2,\ldots,x_N)
dx_1\cdots dx_N
\nonumber\\
&= \frac{1}{N!}\int_{-\infty}^s\dots\int_{-\infty}^s
\mathbb{P}_{2,N}(x_1,x_2,\ldots,x_N)dx_1\cdots dx_N, 
\label{def:F_{2,N}(s)}
\end{align}
and the Tracy-Widom distribution function 
$F_2(s)$ is obtained as a scaling limit of $F_{2,N}(s)$
\cite{Forrester,TracyWidom1}: 
\begin{equation}
F_2(s)= \lim_{N\to\infty}F_{2,N}\left(
\sqrt{2N}+\frac{s}{\sqrt{2}N^{1/6}}\right).
\end{equation}

As discussed in \cite{TracyWidom2}, the distribution function 
$F_{2,N}(s)$ is related to the fourth Painlev\'e equation
\begin{equation}
\frac{d^2 w}{dz^2}=\frac{1}{2w}\left(\frac{dw}{dz}\right)^2
+\frac{3}{2}w^3+4zw^2+2\left(z^2-\alpha\right)w+\frac{\beta}{w}.
\label{P4}
\end{equation}
\begin{thm}[Tracy-Widom \cite{TracyWidom2}]
\label{thm:F2N}
$R(s):= \left\{\log F_{2,N}\right\}'$ is given explicitly by
\begin{equation}
F_{2,N}(s)=\exp\left[-\int_s^{\infty}
\left\{
Nw-\frac{z^2}{2}w-\frac{z}{2}w^2-\frac{w^3}{8}+\frac{1}{8w}
\left(\frac{dw}{dz}\right)^2
\right\}dz\right],
\label{PainleveExpression_F2N}
\end{equation}
where $w=w(z)$ is a solution to the Painlev\'e IV equation
with the parameters $\alpha=2N-1$ and $\beta=0$, 
satisfying the boundary condition 
\begin{equation}
w(s)\sim 0 \mbox{ \ as \ } s\to +\infty.
\label{boundaryCondition:ClarksonMcLeod}
\end{equation}
\end{thm}

The proof of Theorem \ref{thm:F2N} in \cite{TracyWidom2} 
is based on the Fredholm determinantal expression, 
\begin{equation}
F_{2,N}(s)  :=\mathbb{P}_{2,N}(\lambda_{\mathrm{max}}\leq s)
=\det\left[I-K_{2,N}\right]_{L^2(s,\infty)}.
\label{FredholmDet_finiteN}
\end{equation}
The kernel of the integral operator $K_{2,N}$ is given by
\begin{equation}
K_{2,N}(x,z) =
\frac{e^{-(x^2+z^2)/2}}{2^N\sqrt{\pi} (N-1)!}
\cdot
\frac{H_N(x)H_{N-1}(z)-H_{N-1}(x)H_N(z)}{x-z},
\end{equation}
where $H_k(x)$ is the $k$-th Hermite polynomials. 

In what follows, we present an alternative proof based on 
Hirota's bilinear method in soliton theory 
\cite{ASBook,HirotaBook,NimmoZhao}. 
The solutions to many soliton equations can be expressed in terms of
determinants with Wronski-type or Gram-type structure, 
and the bilinear forms of the equations are reduced to 
algebraic identities of determinants
\cite{FreemanNimmo,HirotaBook,HOS,NimmoZhao}. 
We show that the function $F_{2,N}(s)$ can be obtained from
``two-directional'' Wronskian solutions \cite{HOS} to the nonlinear
Schr\"odinger (NLS) equation
\begin{equation}
i\frac{\partial q}{\partial t}
= \frac{\partial^2 q}{\partial x^2}+2|q|^2 q.
\label{NLS_cplx}
\end{equation}
We also consider the asymptotic behavior of $F_{2,N}(s)$ as
$s\to\pm\infty$, and its relation to the 
``Clarkson-McLeod solution'' \cite{BCHM,ClarksonMcLeod,ItsKapaev} 
to the Painlev\'e IV equation.

\section{Nonlinear Sch\"odinger equation and the Painlev\'e IV}
In this section, we review Hirota bilinear formulation of the 
NLS equation \eqref{NLS_cplx}, 
and its similarity reduction to the Painlev\'e IV equation 
\eqref{P4}.

The relation between NLS-type equations and the Painlev\'e IV 
\eqref{P4} has been discussed from various perspectives 
\cite{ARS,ASBook,BCH,Can,JM,KakeiKikuchi,Kitaev,Nakach}. 
Among these works, our approach is based 
mainly on the work of Jimbo and Miwa \cite{JM}, in which they 
introduced the ``$\tau$-function'' in a systematic manner.

Hereafter we will forget the complex structure and consider the 
following coupled system of nonlinear partial differential equations:
\begin{equation}
-\frac{\partial q}{\partial t_2}
+\frac{\partial^2 q}{\partial t_1^2}+2q^2r =0, \quad
\frac{\partial r}{\partial t_2}
+\frac{\partial^2 r}{\partial t_1^2}+2qr^2 =0.
\label{NLS}
\end{equation}
Setting $r=\bar{q}$, $t_1 = x$, $t_2 = it$ recovers the NLS equation \eqref{NLS_cplx}.
It is well-known that the NLS equation has infinite number of 
conserved currents \cite{ASBook}. 
The first three of the conserved currents are given as
\begin{align}
\mathcal{J}_1 &= qr,\\
\mathcal{J}_2 &= \frac{\partial q}{\partial t_1}r-q\frac{\partial r}{\partial
 t_1},\\
\mathcal{J}_3 &= \frac{\partial^2 q}{\partial t_1^2}r
+q\frac{\partial^2 r}{\partial t_1^2}
-2\frac{\partial q}{\partial t_1}\frac{\partial r}{\partial t_1}
+2q^2r^2, 
\end{align}
which satisfy the continuity equations
\begin{equation}
\frac{\partial\mathcal{J}_1}{\partial t_2}
=\frac{\partial\mathcal{J}_2}{\partial t_1},\quad
\frac{\partial\mathcal{J}_2}{\partial t_2}
=\frac{\partial\mathcal{J}_3}{\partial t_1}.
\end{equation}
These equations guarantee the existence of ``$\tau$-function'' 
of the NLS equation \eqref{NLS} that satisfies \cite{JM} 
\begin{equation}
\frac{\partial^2}{\partial t_1^2}\log\tau =\mathcal{J}_1,\quad
\frac{\partial^2}{\partial t_1\partial t_2}\log\tau =\mathcal{J}_2,\quad
\frac{\partial^2}{\partial t_2^2}\log\tau =\mathcal{J}_3.
\label{def:tau}
\end{equation}
We introduce $\rho_1=\rho_1(t_1,t_2)$ and $\rho_2=\rho_2(t_1,t_2)$ as 
\begin{equation}
\rho_1(t_1,t_2) = \tau(t_1,t_2) q(t_1,t_2), \quad
\rho_2(t_1,t_2) = \tau(t_1,t_2) r(t_1,t_2), 
\end{equation}
which is the same dependent-variable transformation used in
\cite{HirotaBook,HOS} to bilinearize \eqref{NLS}.
Then one can rewrite the equations \eqref{NLS} and \eqref{def:tau}
as bilinear differential equations of Hirota-type:
\begin{align}
& (D_1^2+D_2)\tau\cdot\rho_1=0,
\label{bilinearEq2}\\
& (D_1^2+D_2)\rho_2\cdot\tau=0,
\label{bilinearEq3}\\
& D_1^2 \tau\cdot\tau = 2\rho_1\rho_2,
\label{bilinearEq1}\\
& D_1D_2\tau\cdot\tau = 2D_1\rho_1\cdot\rho_2,
\label{bilinearEq4}\\
& D_2^2\tau\cdot\tau = 2D_1^2\rho_1\cdot\rho_2, 
\label{bilinearEq5}
\end{align}
where we have used the Hirota differential operators 
$D_1$, $D_2$ defined by
\begin{equation}
D_1^m D_2^n f(t_1,t_2)\cdot g(t_1,t_2)
=\left.\left(
\frac{\partial}{\partial t_1}-\frac{\partial}{\partial t'_1}\right)^m 
\left(\frac{\partial}{\partial t_2}-\frac{\partial}{\partial t'_2}\right)^n
f(t_1,t_2)g(t'_1,t'_2)\right|_{
\begin{aligned}
\mbox{\scriptsize$t'_1$}&\mbox{\scriptsize$=t_1$}\\[-2mm]
\mbox{\scriptsize$t'_2$}&\mbox{\scriptsize$=t_2$}
\end{aligned}}.
\end{equation}
Straightforward calculation shows
\begin{prop}
\label{prop:similarity:tau_rho1_rho2}
Let $\tau(t_1,t_2)$, $\rho_1(t_1,t_2)$, $\rho_1(t_1,t_2)$ be 
a solution to the bilinear equations 
\eqref{bilinearEq2}--\eqref{bilinearEq5}. 
Let $\lambda$, $c_1$, $c_2$ be constants and assume $\lambda\neq 0$.
Define $\tilde{\tau}(t_1,t_2)$, $\tilde{\rho}_1(t_1,t_2)$, 
$\tilde{\rho}_1(t_1,t_2)$ as 
\begin{align}
\tilde{\tau}(t_1,t_2) &=\lambda^{-1+(c_1+c_2)/2}
\tau(\lambda t_1,\lambda^2 t_2),\\ 
\tilde{\rho}_1(t_1,t_2)&=\lambda^{c_1}\rho_1(\lambda t_1,\lambda^2 t_2),\\
\tilde{\rho}_2(t_1,t_2)&=\lambda^{c_2}\rho_2(\lambda t_1,\lambda^2 t_2).
\end{align}
Then $\tilde{\tau}(t_1,t_2)$, $\tilde{\rho}_1(t_1,t_2)$, 
$\tilde{\rho}_1(t_1,t_2)$ also solve the equations 
\eqref{bilinearEq2}--\eqref{bilinearEq5}. 
\end{prop}
Motivated from Proposition \ref{prop:similarity:tau_rho1_rho2}, 
we now impose similarity conditions on $\tau(t_1,t_2)$,
$\rho_1(t_1,t_2)$, $\rho_1(t_1,t_2)$:
\begin{align}
\tau(t_1,t_2) &=\lambda^{-1+(c_1+c_2)/2}\tau(\lambda t_1,\lambda^2 t_2),
\label{similarity:tau}\\
\rho_1(t_1,t_2)&=\lambda^{c_1}\rho_1(\lambda t_1,\lambda^2 t_2),
\label{similarity:rho1}\\
\rho_2(t_1,t_2)&=\lambda^{c_2}\rho_2(\lambda t_1,\lambda^2 t_2),
\label{similarity:rho2}
\end{align}
for all $\lambda\in\mathbb{C}^{*}$, where $c_1$ and $c_2$ are constants.
Then $q(t_1,t_2)$ and $r(t_1,t_2)$ satisfy 
\begin{align}
q(t_1,t_2)&=\lambda^{1+(c_1-c_2)/2}q(\lambda t_1,\lambda^2 t_2),
\label{similarity:q}\\
r(t_1,t_2)&=\lambda^{1+(-c_1+c_2)/2}r(\lambda t_1,\lambda^2 t_2).
\label{similarity:r}
\end{align}
\begin{prop}[cf. \cite{JM}]
\label{Prop:sigma-form_P4}
Suppose $\tau(t_1,t_2)$, $q(t_1,t_2)$ and $r(t_1,t_2)$ satisfy
the differential equations \eqref{NLS}, \eqref{def:tau}, 
and the similarity conditions \eqref{similarity:tau}, 
\eqref{similarity:q}, \eqref{similarity:r}. 
Then one can introduce $F(s)$ as 
\begin{equation}
\tau(t_1,t_2)=(2\sqrt{-t_2})^{(2-c_1-c_2)/2}F(s), \quad
s=\frac{-t_1}{2\sqrt{-t_2}}.
\label{def:F(s)}
\end{equation}
Furthermore, if we define $H(s)$ as 
\begin{equation}
H(s)=\frac{d}{ds}\log F(s)=\frac{F'(s)}{F(s)}, 
\label{def:G(s)}
\end{equation}
then $H(s)$ solves the following differential equation,
\begin{equation}
(H'')^2-4 (H-s H')^2
+4H'\left\{(H')^2-(c_1-c_2) H'+2(c_1+c_2-2)\right\}=0.
\label{sigma-form:H:c1c2}
\end{equation}
\end{prop}
\begin{proof}
Setting $\lambda=1/(2\sqrt{-t_2})$ and 
$F(s)=\tau(-s,-1/4)$ in \eqref{similarity:tau},  
we obtain \eqref{def:F(s)}. Similarly we obtain
\begin{equation}
\begin{aligned}
q(t_1,t_2)&= (2\sqrt{-t_2})^{(-2-c_1+c_2)/2}u(s),\\
r(t_1,t_2)&= (2\sqrt{-t_2})^{(-2+c_1-c_2)/2}v(s)
\end{aligned}
\label{def:uv}
\end{equation}
by setting $u(s)=q(-s,-1/4)$ and $v(s)=r(-s,-1/4)$.
Substituting \eqref{def:uv} for \eqref{NLS}, we have
\begin{equation}
\begin{aligned}
u''&=2 s u'+2 u^2 v +(c_1-c_2+2)u,\\
v''&=-2 s v'+2 u v^2+(c_1-c_2-2)v.
\end{aligned}
\end{equation}
In the same manner, substituting \eqref{def:F(s)} and 
\eqref{def:uv} for \eqref{def:tau}, we obtain
\begin{equation}
\begin{aligned}
&H'+uv=0, \quad
sH'+H+\frac{1}{2}(u'v-uv')=0,\\
&4s^2H'+12 s H+u''v -2 u'v'+uv''-2 u^2 v^2+4(c_1+c_2-2)=0.
\end{aligned}
\label{similarity_NLS_CQ}
\end{equation}
Eliminating $u$, $v$ from the equations \eqref{similarity_NLS_CQ}, 
one obtains \eqref{sigma-form:H:c1c2}.
\end{proof}

Denote as $2(\theta_{\infty}\pm\theta_0)$ the roots of 
a quadratic polynomial $x^2-(c_1-c_2)x+2(c_1+c_2-2)$. 
Then \eqref{sigma-form:H:c1c2} is written as 
\begin{equation}
\left(H''\right)^2
-4\left(sH'-H\right)^2
+4H'\left(H'-2\theta_{\infty}-2\theta_0\right)
\left(H'-2\theta_{\infty}+2\theta_0\right)=0.
\label{sigma-form:H:theta}
\end{equation}
Setting $\sigma=H-2t\left(\theta_{\infty}+\theta_0\right)$, 
we obtain the $\sigma$-form of the Painlev\'e IV equation \cite{JM}:
\begin{equation}
\left(\sigma''\right)^2
-4\left(s\sigma'-\sigma\right)^2
+4\sigma'\left(\sigma'+4\theta_0\right)
\left(\sigma'+2\theta_{\infty}+2\theta_0\right)=0,
\end{equation}
(cf. \cite{BCH,Clarkson,ForresterWitte,Okamoto}).

\begin{prop}[\cite{BCH,JM}]
Define $Q=Q(s)$, $P=P(s)$ as 
\begin{align}
Q(s) &=
 -\frac{\sigma'(s)}{2}=-\frac{H'(s)}{2}+\theta_{\infty}+\theta_0,
\label{def:Q(s)}\\
P(s) &=
\frac{-\sigma''(s)-2\left\{s\sigma'(s)-\sigma(s)\right\}}{2
\left\{\sigma'(s)+2\theta_{\infty}+2\theta_0\right\}}
=\frac{-H''(s)-2\left\{sH'(s)-H(s)\right\}}{2H'(s)}.
\label{def:P(s)}
\end{align}
Then $H$ can be expressed in terms of $Q$ and $P$: 
\begin{equation}
H=\frac{2}{P}Q^2-\left(P+2s+\frac{4\theta_0}{P}\right)Q
+\left(\theta_0+\theta_{\infty}\right)\left(P+2s\right), 
\label{Hamiltonian}
\end{equation}
and we have the Hamiltonian system associated with the Painlev\'e IV:
\begin{align}
\frac{dQ}{ds}
&=P\frac{\partial H}{\partial P}
=-\frac{2}{P}Q^2+\left(-P+\frac{4\theta_0}{P}\right)Q+(\theta_0+\theta_{\infty})P,
\label{HamiltonianEq1}\\
\frac{dP}{ds}&=-P\frac{\partial H}{\partial Q}
=-4Q+P^2+2s P+4\theta_0.
\label{HamiltonianEq2}
\end{align}
\end{prop}
\begin{proof}
Straightforward calculation shows that 
$P$ of \eqref{def:P(s)} satisfies 
\begin{equation}
H' P^2+H'' P -\left(H'-2\theta_0-2\theta_{\infty}\right)
\left(H'+2\theta_0-2\theta_{\infty}\right)=0.
\label{QuadEq_HP}
\end{equation}
The differential equation 
\begin{equation}
\frac{dQ}{ds}=
-\frac{2}{P}Q^2+\left(-P+\frac{4\theta_0}{P}\right)Q+(\theta_0+\theta_{\infty})P
\label{difEq:dQ/ds}
\end{equation}
follows from \eqref{QuadEq_HP}. 
The relation \eqref{Hamiltonian} follows from 
\eqref{def:Q(s)}, \eqref{def:P(s)} and \eqref{difEq:dQ/ds}. 
Differentiating \eqref{Hamiltonian}, one obtains \eqref{HamiltonianEq2}.
\end{proof}
The following proposition can be verified by straightforward calculation.
\begin{prop} 
Suppose $H=H(s)$ satisfies \eqref{sigma-form:H:theta} and 
define $P(s)$ as \eqref{def:P(s)}. 
Then $w=P(z)$ satisfies the Painlev\'e IV \eqref{P4} with 
$\alpha=2\theta_{\infty}-1$, $\beta=-8\theta_0^2$.
\end{prop}
We remark that 
the Hamiltonian \eqref{Hamiltonian} can be written in 
terms of $w(s)=P(s)$:
\begin{equation}
H(s)=\theta_{\infty}w
-\frac{s^2}{2}w-\frac{s}{2}w^2-\frac{1}{8}w^3
+\frac{(w')^2}{8w}-\frac{2 \theta_0^2}{w}
+2\theta_{\infty}s.
\label{H_intermsof_P}
\end{equation}

\section{Two-directional Wronskian and GUE}
We start considering a class of solution of the bilinear
equations \eqref{bilinearEq2}--\eqref{bilinearEq5}, which is expressed in terms of 
``two-directional'' Wronski determinant:
\begin{equation}
\tau_N(t_1,t_2)=\det\left[\partial_1^{i+j-2} f(t_1,t_2)\right]_{i,j=1}^N.
\label{2D-Wr}
\end{equation}
Hereafter we use the abbreviation $\partial_1=\partial/\partial t_1$, 
$\partial_2=\partial/\partial t_2$.

\begin{thm}[cf. \cite{HOS}]
\label{thm:2d-Wr}
Suppose the function $f(t_1,t_2)$ satisfies
\begin{equation}
\partial_2 f(t_1,t_2)=\partial_1^2 f(t_1,t_2).
\label{DispersionRelation}
\end{equation}
Then the $\tau$-function \eqref{2D-Wr} gives a class of solutions 
for the bilinear equations \eqref{bilinearEq2}--\eqref{bilinearEq5} with 
$\tau=\tau_N(t_1,t_2)$, $\rho_1=\tau_{N+1}(t_1,t_2)$ 
$\rho_2=\tau_{N-1}(t_1,t_2)$. 
\end{thm}
The equations \eqref{bilinearEq2}, \eqref{bilinearEq3}, \eqref{bilinearEq1}
are reduced cases of the equations (4.10a), (4.10b), (4.10c) of \cite{HOS}, 
respectively. 
One can find a proof in \S 4 of \cite{HOS} and we omit here.
We shall give a proof for the remaining two equations in Appendix B.

Here we choose $f(t_1,t_2)$ as
\begin{equation}
f(t_1,t_2)= \int_{-\infty}^0 e^{kt_1+k^2 t_2}dk.
\label{self-similar_f}
\end{equation}
This function admits self-similarity as 
\begin{equation}
f(\lambda t_1,\lambda^2 t_2) = 
\int_{-\infty}^0 e^{k\lambda t_1+k^2\lambda^2 t_2}dk
= \int_{-\infty}^0
e^{kt_1+k^2 t_2}\frac{dk}{\lambda}
=\lambda^{-1}f(t_1,t_2)
\end{equation}
for $\lambda>0$, 
and hence the two-directional Wronskian \eqref{2D-Wr} also 
has self-similarity:
\begin{equation}
\tau_N(\lambda t_1,\lambda^2 t_2) 
= \lambda^{-N^2} \tau_N(t_1,t_2).
\end{equation}
Thus the scaling exponent for $\rho_1(t_1,t_2)=\tau_{N+1}(t_1,t_2)$ is 
$c_1=(N+1)^2$, and that for $\rho_2(t_1,t_2)=\tau_{N-1}(t_1,t_2)$ is
$c_2=(N-1)^2$. 
It follows from Proposition \ref{Prop:sigma-form_P4} that 
\begin{equation}
F(s)=\left(2\sqrt{-t_2}\right)^{-(2-c_1-c_2)/2}\tau_N(t_1,t_2)
=\left(2\sqrt{-t_2}\right)^{N^2}\tau_N(t_1,t_2)
\label{F_to_tauN}
\end{equation}
depends only on $s=-t_1/(2\sqrt{-t_2})$. 
Since $c_1=(N+1)^2$, $c_2=(N-1)^2$, 
$H(s)$ of \eqref{def:G(s)} satisfies the $\sigma$-form Painlev\'e IV 
\eqref{sigma-form:H:theta} with 
$\theta_0=0$, $\theta_{\infty}=N$:
\begin{equation}
(H'')^2-4 (H-s H')^2+4H'\left(H'-2N\right)^2=0.
\label{diffEq:G_F2N}
\end{equation}

The two-directional Wronskian $\tau_N$ under consideration 
coincides with the distribution function $F_{2,N}(s)$ as shown below.
For this purpose, we assume $t_2<0$ and 
substitute $f(t_1,t_2)$ of \eqref{self-similar_f} into 
\eqref{2D-Wr}:
\begin{align}
\tau_N &=\det\left[
\int_{-\infty}^{0}k^{i+j-2}e^{kt_1+k^2 t_2}dk
\right]_{i,j=1}^N
\nonumber\\
&=\frac{e^{Ns^2}}{(-t_2)^{N^2/2}}\det\left[
\int_{-\infty}^{s}(k-s)^{i+j-2}e^{-k^2}dk
\right]_{i,j=1}^N
\nonumber\\
&=\frac{e^{Ns^2}}{(-t_2)^{N^2/2}}\det\left[
\int_{-\infty}^{s}k^{i+j-2}e^{-k^2}dk
\right]_{i,j=1}^N.
\label{tau_detphiphi}
\end{align}

On the other hand, applying an integral identity 
\cite{Andreief,deBruijn,TracyWidom3}
\begin{align}
&\idotsint\det\left[\phi_j(x_k)\right]_{j,k=1,\dots,N}
\det\left[\psi_j(x_k)\right]_{j,k=1,\dots,N}dx_1\cdots x_N
\nonumber\\
&= N! \det\left[\int\phi_j(x)\psi_k(x)dx\right]_{j,k=1,\dots,N}
\label{deBruijnIdentity}
 \end{align}
to \eqref{def:F_{2,N}(s)}, one can rewrite $F_{2,N}(s)$ as 
\begin{equation}
F_{2,N}(s)=\frac{2^{N(N-1)/2}\pi ^{-N/2}}{\prod_{j=1}^N j!}
\det\left[\int_{-\infty}^s
x^{j+k-2}e^{-x^2}dx\right]_{j,k=1,\dots,N}.
\label{detrep_F2N}
\end{equation}
Comparing \eqref{F_to_tauN}, \eqref{tau_detphiphi}, and
\eqref{detrep_F2N}, we obtain
\begin{equation}
F(s) = 2^{N(N+1)/2}\pi^{N/2}e^{Ns^2}
\left\{\prod_{j=1}^{N-1}(j!)\right\}F_{2,N}(s), 
\end{equation}
and thus
\begin{equation}
H(s)=2Ns + R(s), \quad 
R(s):=\frac{d}{ds}\log F_{2,N}(s).
\label{def:R(s)}
\end{equation}
It follows from \eqref{H_intermsof_P} and \eqref{def:R(s)} that 
\begin{equation}
\frac{d}{ds}\log F_{2,N}(s) = 
Nw-\frac{s^2}{2}w-\frac{s}{2}w^2-\frac{w^3}{8}
+\frac{\left(w'\right)^2}{8w}. 
\end{equation}
Integrating once, we obtain \eqref{PainleveExpression_F2N}.

\section{Asymptotic behavior}
Real solutions to the Painlev\'e IV equation
in the case $\beta=0$ and real $\alpha$ with the boundary condition 
\eqref{boundaryCondition:ClarksonMcLeod} are called 
``Clarkson-McLeod solution'' \cite{ClarksonMcLeod,BCHM,ItsKapaev}.
\begin{thm}[Clarkson-McLeod solution \cite{Abdullaye,BCHM,Lu}]
\label{thm:Clarkson-McLeodSolution}
In the case $\beta=0$ and real $\alpha$, 
there exists a unique one-parameter family $w(x;k^2)$ 
($k^2\in\mathbb{R}$) of real solutions of 
the Painlev\'e IV equation \eqref{P4} determined by asymptotic condition
\begin{equation}
w(s) 
 \sim k^2 2^{(\alpha+2)/2}s^{\alpha-1}e^{-s^2}
\mbox{ \ as \ } s\to +\infty.
\end{equation}
Moreover, in the case $\beta=0$ and real $\alpha$, 
any solutions of the Painlev\'e IV equation \eqref{P4} with the 
boundary condition \eqref{boundaryCondition:ClarksonMcLeod} 
coincides with $w(x;k^2)$ for some constant $k^2\in\mathbb{R}$.
\end{thm}

The parameters $k$, $\alpha$ 
for the solution of the Painlev\'e IV associated with $F_{2,N}(s)$ 
have been evaluated by Tracy and Widom 
(\cite{TracyWidom2}, p.55).
\begin{thm}[Tracy-Widom \cite{TracyWidom2}]
\label{thm:asymptotics}
$R(s):= \left\{\log F_{2,N}\right\}'$ is given explicitly by 
\eqref{PainleveExpression_F2N} 
where $w=w(z;k^2)$ is the Clarkson-McLeod solution (Theorem 
\ref{thm:Clarkson-McLeodSolution}) to the Painlev\'e IV, 
with the parameters $\alpha=2N-1$, $\beta=0$, 
and $k^2=2^{-3/2}\pi^{-1/2}/(N-1)!$.
\end{thm}
Since no explicit proof of Theorem \ref{thm:asymptotics} 
has been written in \cite{TracyWidom2},  
we give a proof based on the determinant expression 
\eqref{detrep_F2N}. 
Toward this purpose, we consider two integrals
\begin{equation}
I_n:= \int_{-\infty}^{\infty} x^{n}e^{-x^2}dx,
\quad
J_n:= \int_{s}^{\infty} x^{n}e^{-x^2}dx
\quad (n=0,1,2,\ldots)
\end{equation}
and evaluate a determinant 
\begin{equation}
T_N(s) := \det\left[
\int_{-\infty}^s x^{i+j-2}e^{-x^2}dx
\right]_{i,j=1}^N
=\det\left[I_{m+n-2}-J_{m+n-2}\right]_{i,j=1}^N
\label{def:T_N(s)}
\end{equation}
for large $s$. Here $J_n$ can be represented in terms of 
the incomplete gamma function, 
\begin{equation}
\Gamma(a,z) = \int_z^{\infty} t^{a-1}e^{-t}dt, 
\end{equation}
as
\begin{equation}
J_n = \frac{1}{2}\Gamma\left((n+1)/2,s^2\right).
\end{equation}
Using the asymptotic series \cite{AbramowitzStegun}
\begin{equation}
\Gamma(a,z)\sim
z^{a-1}e^{-z}\left[1+\frac{a-1}{z}
+\cdots\right]
\quad \left(
\mbox{as }z\to\infty\mbox{ in }
\left|\mathrm{arg}z\right|<\frac{3\pi}{2}
\right), 
\end{equation}
we can estimate asymptotic behavior of $J_n(s)$ as
\begin{equation}
J_n \sim \frac{1}{2}s^{n-1}e^{-s^2}
\left[1+\frac{n-1}{2s^2}+\cdots\right]
\quad\left(\mbox{as }s\to\infty\right).
\label{Jn_asymptotics}
\end{equation}
Applying \eqref{Jn_asymptotics} to \eqref{def:T_N(s)}, 
we have
\begin{equation}
T_N(s)\sim \det\left[I_{i+j-2}\right]_{i,j=1}^N
-\frac{1}{2}\det\left[I_{i+j-2}\right]_{i,j=1}^{N-1}
s^{2N-3}e^{-s^2}
\quad (\mbox{as }s\to\infty).
\end{equation}
It follows from \eqref{normalization} and 
\eqref{deBruijnIdentity} that
\begin{equation}
\det\left[I_{i+j-2}\right]_{i,j=1}^N
= \frac{1}{N!}Z_N
= 2^{-N(N-1)/2}\pi ^{N/2}\prod_{j=1}^{N-1} j!, 
\end{equation}
and thus
\begin{equation}
T_N(s)\sim \frac{Z_N}{N!}
-\frac{1}{2}\frac{Z_{N-1}}{(N-1)!}
s^{2N-3}e^{-s^2}
\quad (\mbox{as }s\to\infty).
\label{leadingTerm_T_N}
\end{equation}
Applying \eqref{detrep_F2N}, \eqref{def:T_N(s)} and 
\eqref{leadingTerm_T_N} to \eqref{def:R(s)}, 
we obtain
\begin{align}
R_N(s) &= \frac{d}{ds}\log T_N(s) \nonumber\\
&\sim N \frac{Z_{N-1}}{Z_N}s^{2N-2}e^{-s^2} \quad (\mbox{as }s\to\infty)
\nonumber\\
&= \frac{2^{N-1}}{\pi^{1/2}(N-1)!}s^{2N-2}e^{-s^2}.
\end{align}

\begin{prop} Higher order expansion of $R(s)$ is given as
\begin{align}
R(s)\sim & \frac{2^{N-1}}{\pi^{1/2}(N-1)!}s^{2N-2}e^{-s^2}
\nonumber\\
&\times \left[
1-\frac{(N-1)(N-3)}{2s^2}+
\frac{(N-1)(N-2)(N^2-7N+15)}{8s^4}+\cdots
\right] \quad (\mbox{as }s\to\infty).
\label{R(s)_higher_asymptotics}
\end{align}
\end{prop}
\begin{proof} 
We remark that $R(s)$ satisfies the differential equation 
\begin{equation}
\left(R''\right)^2-4\left(R-sR'\right)^2+4\left(R'\right)^2
\left(R'+2N\right)=0.
\label{R_N:sigma-form}
\end{equation}
Consider the asymptotic expansion 
\begin{equation}
R(s) \sim s^{2N-2}e^{-s^2}\left(
c_0 + c_1s^{-1} + c_2s^{-2} + c_3s^{-3}+\cdots
\right) 
\label{R_N:asympt_exp}
\end{equation}
with $c_0=2^{N-1}\pi^{-1/2}/(N-1)!$. 
Inserting \eqref{R_N:asympt_exp} into \eqref{R_N:sigma-form}, 
one can obtain $c_1$, $c_2$, $c_3$, $\ldots$, recursively, 
and the desirous result follows.
\end{proof}

\begin{proof}[Proof of Theorem \ref{thm:asymptotics}]
{}From \eqref{def:P(s)} and \eqref{def:R(s)}, we have
\begin{equation}
w(s) = \frac{-R''-2sR'+2R}{2\left(R'+2N\right)}.
\label{w_and_R}
\end{equation}
Inserting the asymptotic series \eqref{R(s)_higher_asymptotics} to
\eqref{w_and_R}, we have the desirous result.
\end{proof}

\section*{Appendices} 
\def\theequation{A.\arabic{equation}}
\setcounter{equation}{0}
\subsection*{Appendix A: Bilinear-theoretic derivation of \eqref{normalization}}
In \cite{MehtaBook}, the proof of the integral formula
\eqref{normalization} is given by using the Selberg integral.
Here we give another proof based on the determinant expression
\begin{equation}
Z_N = N! \det\left[I_{i+j-2}\right]_{i,j=1}^N,\quad
I_n = \int_{-\infty}^{\infty}x^n e^{-x^2}dx,
\end{equation}
which is a consequence of the identity \eqref{deBruijnIdentity}.
If $n$ is odd then $I_n=0$, and if $n$ is even, 
the integral $I_n$ is a special value of the Gamma function:
\begin{equation}
I_{2m} 
= \int_0^{\infty}t^{(2m-1)/2}e^{-t}dt
= \Gamma\left(m+\frac{1}{2}\right)
= \sqrt{\pi}\prod_{j=1}^m \left(m-\frac{1}{2}\right).
\end{equation}

We recall the Desnanot-Jacobi identity for determinants 
\cite{Bressoud,HirotaBook}
(or ``Lewis Carroll identity'', ``Dodgson condensation''). 
We consider a determinant $D$ of degree $N$ and its minors. 
We denote by $D\begin{bmatrix}i\\ j\end{bmatrix}$ the minor determinant 
obtained by deleting the $i$th row and the $j$th column, and 
by $D\begin{bmatrix}i_1 & i_2\\ j_1 & j_2\end{bmatrix}$ 
the minor obtained by deleting the $i_1$th and $i_2$th rows 
and the $j_1$th and $j_2$th columns.
\begin{thm}[Desnanot-Jacobi identity \cite{Bressoud,HirotaBook}]
For $i_1$, $i_2$, $j_1$, $j_2$ that satisfy 
$1\leq i_1<i_2\leq N$ and $1\leq j_1<j_2\leq N$, the following 
identity holds: 
\begin{equation}
D\cdot D\begin{bmatrix}i_1 & i_2\\ j_1 & j_2\end{bmatrix}
=D\begin{bmatrix}i_1\\ j_1\end{bmatrix}
D\begin{bmatrix}i_2\\ j_2\end{bmatrix}-
D\begin{bmatrix}i_1\\ j_2\end{bmatrix}
D\begin{bmatrix}i_2\\ j_1\end{bmatrix}.
\label{Desnanot-Jacobi}
\end{equation}
\end{thm}

For $M=0,1,2,\dots$, $N=1,2,\ldots$, define $A^M_N$ as
\begin{equation}
A^M_N = \det\left[I_{M+i+j-2}\right]_{i,j=1}^N.
\end{equation}
{}From the Desnanot-Jacobi identity \eqref{Desnanot-Jacobi}, 
we have
\begin{equation}
A^M_{N+2}A^{M+2}_{N}=A^M_{N+1}A^{M+2}_{N+1}-\left(A^{M+1}_{N+1}\right)^2.
\end{equation}
Setting $M=2m$ and $N=2n$, we obtain
\begin{equation}
A^{2m}_{2n+2}A^{2m+2}_{2n} = A^{2m}_{2n+1}A^{2m+2}_{2n+1}
\label{recursion1}
\end{equation}
where we have used $A^{2m+1}_{2n+1}=0$. Similarly we have
\begin{align}
&A^{2m}_{2n+1}A^{2m+2}_{2n-1}=
A^{2m}_{2n}A^{2m+2}_{2n}-\left(A^{2m+1}_{2n}\right)^2,
\label{recursion2}\\
&A^{2m-1}_{2n+2}A^{2m+1}_{2n}=
-\left(A^{2m}_{2n+1}\right)^2,
\label{recursion3}\\
& 
0=A^{2m-1}_{2n}A^{2m+1}_{2n}-\left(A^{2m}_{2n}\right)^2.
\label{recursion4}
\end{align}
The values of $A^M_N$ ($M,N=0,1,2,\ldots$) are determined 
by the recursion relations \eqref{recursion1}, \eqref{recursion2}, 
\eqref{recursion3}, \eqref{recursion4}, together with the initial
conditions 
\begin{equation}
 A^M_0 = 1, \quad A^M_1 = I_M \quad (M=0,1,2,\ldots).
\label{inii_cond}
\end{equation}

\begin{prop}
\label{prop:A^M_N}
Define $P(l;m,n)$ as 
\begin{equation}
P(l;m,n) = 
\begin{cases}
1 & (m=0, n=1,2,\ldots)\\
\prod_{i=1}^m \prod_{j=1}^n
\left(l+i+j-\frac{3}{2}\right) & (m,n=1,2,\ldots)
\end{cases}
\end{equation}
Then $A^M_N$ is expressed in terms of $P(l;m,n)$ as
\begin{align}
A^{2m}_{2n} &= \pi^n P(0;m,n)P(1;m,n)\prod_{j=1}^{2n-1}\frac{j!}{2},
\label{A_even_even}\\
A^{2m}_{2n+1} &=
 \pi^{(2n+1)/2}P(0;m,n+1)P(1;m,n)\prod_{j=1}^{2n}\frac{j!}{2}, 
\label{A_even_odd}\\
A^{2m+1}_{2n} &= (-\pi)^n P(1;m,n)^2 \prod_{j=0}^{n-1}\left\{
\frac{(2j+1)!}{2^{2j+1}}\right\}^2, 
\label{A_odd_even}\\
A^{2m+1}_{2n+1} &= 0 \qquad (m,n=0,1,2,\ldots).
\label{A_odd_odd}
\end{align}
\end{prop}
\begin{proof}
It is enough to show that 
\eqref{A_even_even}, \eqref{A_even_odd}, \eqref{A_odd_even} satisfy 
the recursion relations \eqref{recursion1}, \eqref{recursion2}, 
\eqref{recursion3}, \eqref{recursion4}, and the initial
conditions \eqref{inii_cond}, that can be checked by straightforward
calculation.
\end{proof}
The formula \eqref{normalization} follows from a special case of Proposition 
\ref{prop:A^M_N}: 
\begin{equation}
Z_N = N! A^0_N = N!\cdot \pi^{N/2}\prod_{j=1}^{N-1}\frac{j!}{2} 
= 2^{-N(N-1)/2}\pi ^{N/2}\prod_{j=1}^{N} j!.
\end{equation}

\def\theequation{B.\arabic{equation}}
\setcounter{equation}{0}
\subsection*{Appendix B: Proof of Theorem \ref{thm:2d-Wr}}
In this appendix, we show that the two-directional Wronskian 
\eqref{2D-Wr} actually satisfies the bilinear equations 
\eqref{bilinearEq4} and \eqref{bilinearEq5} with 
$\tau=\tau_N(t_1,t_2)$, $\rho_1=\tau_{N+1}(t_1,t_2)$ 
$\rho_2=\tau_{N-1}(t_1,t_2)$. 
We first introduce
the following notation, which is a generalization of 
the abbreviated notation due to Freeman and Nimmo 
\cite{FreemanNimmo,HirotaBook,HOS,NimmoZhao}:
\begin{equation}
\begin{vmatrix}
m_1, m_2, \ldots, m_N\\
n_1, n_2, \ldots, n_N
\end{vmatrix}:=
\det\left[\partial_1^{m_i+n_j} f(t_1,t_2)\right]_{i,j=1}^N.
\end{equation}
The $\tau$-function \eqref{2D-Wr} and its derivatives are 
written as 
\begin{align}
\tau_N &= \begin{vmatrix}
0,\dots,N-1\\
0,\dots,N-1
\end{vmatrix},
\label{DiffRel1}\\
\partial_1\tau_N &
= \begin{vmatrix}
0,\dots,N-1\\
0,\dots,N-2,N
\end{vmatrix} = \begin{vmatrix}
0,\dots,N-2,N\\
0,\dots,N-1
\end{vmatrix},
\label{DiffRel2}\\
\partial_1^2\tau_N 
&= \begin{vmatrix}
0,\dots,N-2,N\\
0,\dots,N-2,N
\end{vmatrix}
=\begin{vmatrix}
0,\dots,N-3,N-1,N\\
0,\dots,N-1
\end{vmatrix}
+\begin{vmatrix}
0,\dots,N-2,N+1\\
0,\dots,N-1
\end{vmatrix}.
\label{DiffRel3}
\end{align}
Using \eqref{DispersionRelation}, we have
\begin{align}
\partial_2\tau_N 
&=-\begin{vmatrix}
0,\dots,N-3,N-1,N\\
0,\dots,N-1
\end{vmatrix}
+\begin{vmatrix}
0,\dots,N-2,N+1\\
0,\dots,N-1
\end{vmatrix},
\label{DiffRel5}\\
\partial_1\partial_2\tau_N 
&=-\begin{vmatrix}
0,\dots,N-3,N-1,N\\
0,\dots,N-2,N
\end{vmatrix}
+\begin{vmatrix}
0,\dots,N-2,N+1\\
0,\dots,N-2,N
\end{vmatrix},
\label{DiffRel6}\\
\partial_2^2\tau_N 
&=\begin{vmatrix}
0,\dots,N-3,N-1,N\\
0,\dots,N-3,N-1,N
\end{vmatrix}
\nonumber\\
&\qquad -2\begin{vmatrix}
0,\dots,N-3,N-1,N\\
0,\dots,N-2,N+1
\end{vmatrix}
+\begin{vmatrix}
0,\dots,N-2,N+1\\
0,\dots,N-2,N+1
\end{vmatrix}.
\label{DiffRel7}
\end{align}

Using the differential formulas \eqref{DiffRel2} and the identity 
\eqref{Desnanot-Jacobi}, we obtain
\begin{align}
\left(\partial_1\tau_{N+1}\right)\tau_{N-1} &=
\begin{vmatrix}0,\ldots,N-1,N+1\\ 0,\ldots,N\end{vmatrix}
\begin{vmatrix}0,\ldots,N-2\\ 0,\ldots,N-2\end{vmatrix}
\nonumber\\
&=\begin{vmatrix}0,\ldots,N-2,N+1\\ 0,\ldots,N-2,N\end{vmatrix}
\begin{vmatrix}0,\ldots,N-1\\ 0,\ldots,N-1\end{vmatrix}
\nonumber\\
&\qquad -\begin{vmatrix}0,\ldots,N-2,N+1\\ 0,\ldots,N-1\end{vmatrix}
\begin{vmatrix}0,\ldots,N-1\\ 0,\ldots,N-2,N\end{vmatrix},
\label{proof_blinearEq4_1}
\\
\tau_{N+1}\left(\partial_1\tau_{N-1}\right)&=
\begin{vmatrix}0,\ldots,N\\ 0,\ldots,N\end{vmatrix}
\begin{vmatrix}0,\ldots,N-3,N-1\\ 0,\ldots,N-2\end{vmatrix}
\nonumber\\
&=\begin{vmatrix}0,\ldots,N-3,N-1,N\\ 0,\ldots,N-2,N\end{vmatrix}
\begin{vmatrix}0,\ldots,N-1\\ 0,\ldots,N-1\end{vmatrix}
\nonumber\\
&\qquad -\begin{vmatrix}0,\ldots,N-3,N-1,N\\ 0,\ldots,N-1\end{vmatrix}
\begin{vmatrix}0,\ldots,N-1\\ 0,\ldots,N-2,N\end{vmatrix}
\label{proof_blinearEq4_2}
\end{align}
It is straightforward to show that $\tau=\tau_N$, $\rho_1=\tau_{N+1}$, 
$\rho_2=\tau_{N-1}$ solve \eqref{bilinearEq4} by using 
\eqref{proof_blinearEq4_1}, \eqref{proof_blinearEq4_2}, 
and the differential relations 
\eqref{DiffRel2}, \eqref{DiffRel5}, \eqref{DiffRel6}.

We use the Desnanot-Jacobi identity \eqref{Desnanot-Jacobi} 
again to obtain 
\begin{align}
\left(\partial_1^2\tau_{N+1}\right)\tau_{N-1}&=
\begin{vmatrix}
0,\ldots,N-1,N+1\\ 0,\ldots,N-1,N+1
\end{vmatrix}
\begin{vmatrix}
0,\ldots,N-2\\ 0,\ldots,N-2
\end{vmatrix}
\nonumber\\
&=\begin{vmatrix}
0,\ldots,N-2,N+1\\ 0,\ldots,N-2,N+1
\end{vmatrix}
\begin{vmatrix}
0,\ldots,N-1\\ 0,\ldots,N-1
\end{vmatrix}
\nonumber\\
&\qquad -\begin{vmatrix}
0,\ldots,N-2,N+1\\ 0,\ldots,N-1
\end{vmatrix}
\begin{vmatrix}
0,\ldots,N-1\\ 0,\ldots,N-2,N+1
\end{vmatrix},
\\
\tau_{N+1}\left(\partial_1^2\tau_{N-1}\right)
&=\begin{vmatrix}
0,\ldots,N\\ 0,\ldots,N
\end{vmatrix}
\begin{vmatrix}
0,\ldots,N-3,N-1\\ 0,\ldots,N-3,N-1
\end{vmatrix}
\nonumber\\
&=\begin{vmatrix}
0,\ldots,N-3,N-1,N\\ 0,\ldots,N-3,N-1,N
\end{vmatrix}
\begin{vmatrix}
0,\ldots,N-1\\ 0,\ldots,N-1
\end{vmatrix}
\nonumber\\
&\qquad -\begin{vmatrix}
0,\ldots,N-3,N-1,N\\ 0,\ldots,N-1
\end{vmatrix}
\begin{vmatrix}
0,\ldots,N-1\\ 0,\ldots,N-3,N-1,N
\end{vmatrix},
\\
\left(\partial_1\tau_{N+1}\right)\left(\partial_1\tau_{N-1}\right)
&=\begin{vmatrix}
0,\ldots,N-1,N+1\\ 0,\ldots,N
\end{vmatrix}
\begin{vmatrix}
0,\ldots,N-2\\ 0,\ldots,N-3,N-1
\end{vmatrix}
\nonumber\\
&=\begin{vmatrix}
0,\ldots,N-2,N+1\\ 0,\ldots,N-3,N-1,N
\end{vmatrix}
\begin{vmatrix}
0,\ldots,N-1\\ 0,\ldots,N-1
\end{vmatrix}
\nonumber\\
&\qquad -\begin{vmatrix}
0,\ldots,N-2,N+1\\ 0,\ldots,N-1
\end{vmatrix}
\begin{vmatrix}
0,\ldots,N-1\\ 0,\ldots,N-3,N-1,N
\end{vmatrix}.
\end{align}
These relations give \eqref{bilinearEq5}
with $\tau=\tau_N$, $\rho_1=\tau_{N+1}$, and $\rho_2=\tau_{N-1}$.

\def\theequation{C.\arabic{equation}}
\setcounter{equation}{0}
\subsection*{Appendix C: Bilinearization of the Painlev\'e IV}
As in Proposition \ref{Prop:sigma-form_P4}, 
under the similarity conditions 
\eqref{similarity:tau}--\eqref{similarity:rho2}, 
one can introduce $F(s)$ as \eqref{def:F(s)} and 
$G_1(s)$, $G_2(s)$ as 
\begin{equation}
\rho_1(t_1,t_2) =(2\sqrt{-t_2})^{-c_1}G_1(s), \quad
\rho_2(t_1,t_2) =(2\sqrt{-t_2})^{-c_2}G_2(s).
\end{equation}
Bilinear equations for $F(s)$, $G_1(s)$ and $G_2(s)$ 
follow from \eqref{bilinearEq2}--\eqref{bilinearEq5}:
\begin{align}
& \left(D_s^2 +2s D_s -c_1+c_2-2\right) F\cdot G_1=0,\\
& \left(D_s^2 +2s D_s -c_1+c_2+2\right) G_2\cdot F=0,\\
& D_s^2 F\cdot F = G_1 G_2, \\
& \left(s D_s^2 +\frac{d}{ds}\right)F\cdot F = D_s G_1\cdot G_2,\\
& \left(2s^2 D_s^2+4(c_1+c_2-2)+6s\frac{d}{ds}\right)F\cdot F
= D_s^2 G_1\cdot G_2.
\end{align}
These give a bilinearization of the Painlev\'e IV, however, which are 
different from that used in \cite{HW,KO_P4,NoumiYamada}.

\section*{Acknowledgments}
The author acknowledges Takashi Imamura, Makoto Katori, and 
Tomohiro Sasamoto for discussions on the Tracy-Widom distributions.
Thanks are also due to Kenji Kajiwara and Ralph Willox for information 
on bilinear forms of the Painlev\'e equations. 
This work is partially supported by JSPS Grant-in-Aid for 
Scientific Research No. 23540252.

\end{document}